\documentclass[10pt,conference]{IEEEtran}
\IEEEoverridecommandlockouts
    
\usepackage{cite}
\usepackage{amsthm}
\usepackage{amsmath,amssymb,amsfonts}
\usepackage{graphicx}
\usepackage{textcomp}
\usepackage{xcolor}
\usepackage{booktabs}
\usepackage{array}
\usepackage{multirow}
\usepackage{booktabs}
\usepackage{tabularray}
\usepackage{tcolorbox}
\usepackage{color,soul}
\usepackage[normalem]{ulem}
\usepackage[hidelinks]{hyperref}
\usepackage{balance}
\usepackage[]{mdframed}
\mdfdefinestyle{graybox}{%
    linecolor=black!70,
    outerlinewidth=0.5pt,
    roundcorner=3pt,
    innertopmargin=5pt,
    innerbottommargin=5pt,
    innerrightmargin=5pt,
    innerleftmargin=5pt,
    backgroundcolor=white,
}

\useunder{\uline}{\ul}{}

\makeatletter
\newcommand*{\rom}[1]{\expandafter\@slowromancap\romannumeral #1@}

\usepackage{algorithm}
\usepackage{algpseudocode}

\newtheorem{theorem}{Theorem}[section]

\newtheorem{lemma}[theorem]{Lemma}

\def\BibTeX{{\rm B\kern-.05em{\sc i\kern-.025em b}\kern-.08em
    T\kern-.1667em\lower.7ex\hbox{E}\kern-.125emX}}
\begin{document}

\title{OptLLM: Optimal Assignment of Queries to Large Language Models}

\author{
    \IEEEauthorblockN{Yueyue Liu\IEEEauthorrefmark{2}, Hongyu Zhang\IEEEauthorrefmark{3}*\protect\thanks{*Corresponding author.}, Yuantian Miao\IEEEauthorrefmark{2}, Van-Hoang Le\IEEEauthorrefmark{2}, Zhiqiang Li\IEEEauthorrefmark{4}}
    \IEEEauthorblockA{\IEEEauthorrefmark{2}School of Information and Physical Sciences, The University of Newcastle, Newcastle, Australia}
    \IEEEauthorblockA{\IEEEauthorrefmark{3}School of Big Data and Software Engineering, Chongqing University, Chongqing, China}
    \IEEEauthorblockA{\IEEEauthorrefmark{4}School of Computer Science, Shaanxi Normal University, Shaanxi, China}
    \IEEEauthorblockA{yueyue.liu@uon.edu.au, hyzhang@cqu.edu.cn, sky.miao@newcastle.edu.au, vanhoang.le@uon.edu.au, lizq@snnu.edu.cn}
}

\maketitle

\begin{abstract}

Large Language Models (LLMs) have garnered considerable attention owing to their remarkable capabilities, leading to an increasing number of companies offering LLMs as services. Different LLMs achieve different performance at different costs. A challenge for users lies in choosing the LLMs that best fit their needs, balancing cost and performance.
In this paper, we propose a framework for addressing the cost-effective query allocation problem for LLMs. Given a set of input queries and candidate LLMs, our framework, named OptLLM, provides users with a range of optimal solutions to choose from, aligning with their budget constraints and performance preferences, including options for maximizing accuracy and minimizing cost. 
OptLLM predicts the performance of candidate LLMs on each query using a multi-label classification model with uncertainty estimation and then iteratively generates a set of non-dominated solutions by destructing and reconstructing the current solution. 
To evaluate the effectiveness of OptLLM, we conduct extensive experiments on various types of tasks, including text classification, question answering, sentiment analysis, reasoning, and log parsing. Our experimental results demonstrate that OptLLM substantially reduces costs by 2.40\% to 49.18\% while achieving the same accuracy as the best LLM. Compared to other multi-objective optimization algorithms, OptLLM improves accuracy by 2.94\% to 69.05\% at the same cost or saves costs by 8.79\% and 95.87\% while maintaining the highest attainable accuracy.
\end{abstract}


\section{Introduction}
Large Language Models (LLMs) are increasingly influential due to their ability to analyze and generate natural and programming languages. These models are trained on extensive text data~\cite{kasneci2023chatgpt} and have demonstrated impressive abilities to simulate human linguistic capabilities, posing a significant impact across multiple domains such as natural language processing (NLP)~\cite{ouyang2022training,xie2022explanation,min2022rethinking} and programming tasks~\cite{liu2023your,xia2023automated,xia2023keep}. 

Many companies (such as OpenAI\footnote{\label{openai}https://openai.com}, AI21\footnote{\label{ai21}https://www.ai21.com/}, TogetherAI\footnote{\label{together}https://www.together.ai/},etc.) now offer LLMs as services through public Application Programming Interfaces (APIs). These LLM services have different performance and pricing structures. 
Despite the remarkable achievements and rapid progress of LLMs across various fields, their adoption in practice is challenging. First, utilizing LLMs can lead to considerable costs due to the large number of LLM queries required by a task~\cite{vsakota2023fly}. For instance, using GPT-3 for a simplified customer service system could potentially cost a small business over \$14,400 per month~\cite{gpt3}. 
Second, the performance of an LLM, measured by the percentage of correctly processed queries (referred to as accuracy in this paper), may not be consistent across all tasks. 
Although expansive models such as GPT-4 are expected to wield greater computational strength and deliver superior outcomes, they incur higher costs and do not universally assure satisfactory results for all types of tasks. 
Small LLMs offer economical alternatives for executing less complex tasks, especially in scenarios where the extensive capabilities of larger models are superfluous~\cite{vsakota2023fly}. Moreover, a task that may not be successfully completed by one LLM could potentially be effectively handled by another~\cite{chen2023frugalgpt}.

\begin{figure}
    \centering
    \includegraphics[width=0.95\linewidth]{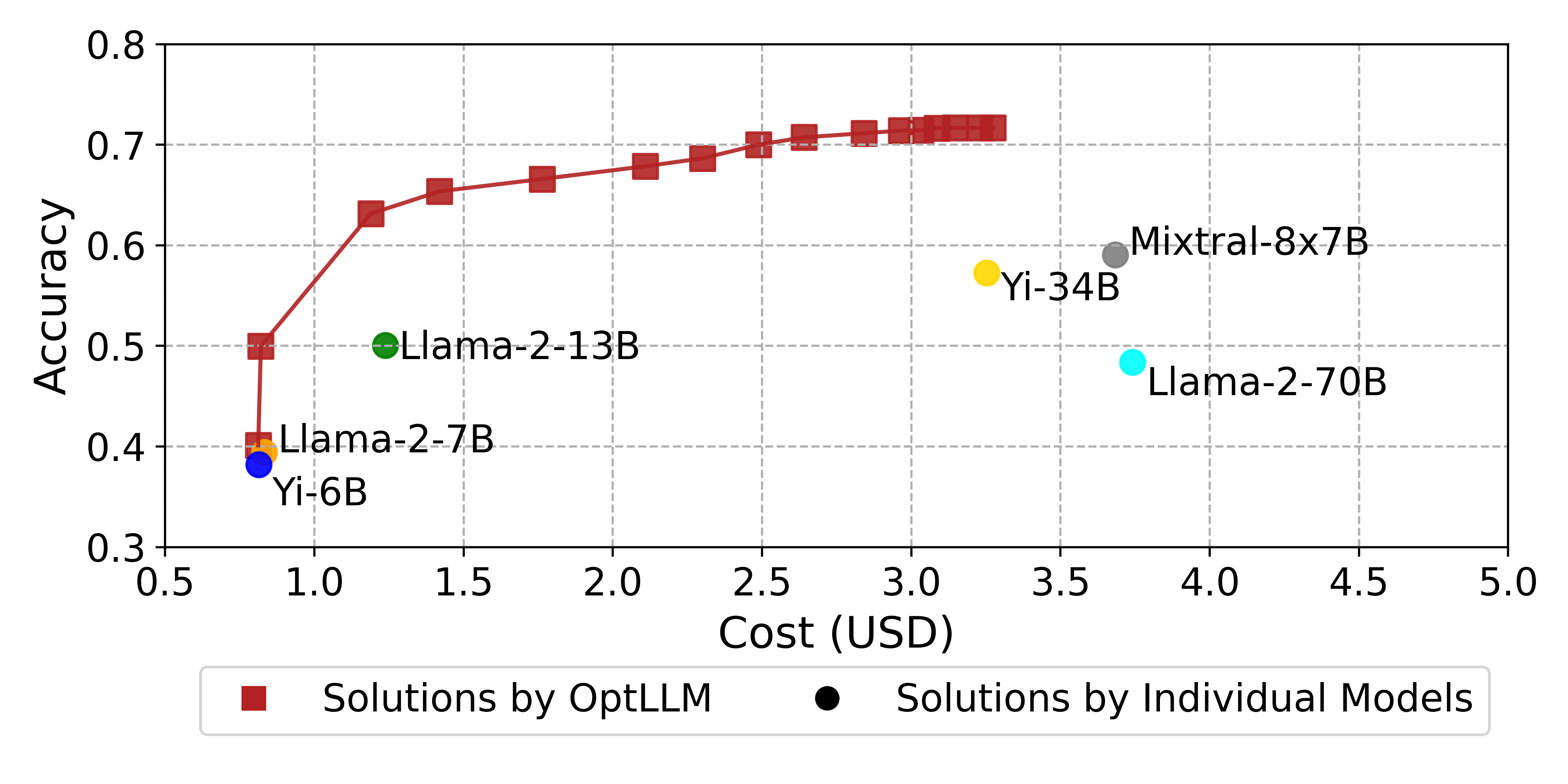}
    \caption{Comparison of using each individual LLM with solutions by OptLLM. OptLLM can provide solutions with higher accuracy and lower cost than the best individual LLM (Mixtral 8x7B) on LLM-based log parsing task} 
    \label{fig_motivation_log}
\end{figure}

Given these challenges, it is imperative to make informed decisions about selecting an appropriate model for each query when leveraging LLMs in order to achieve 
a trade-off between cost and performance expectations. Our OptLLM addresses this issue by predicting the performance of LLMs and optimizing the cost and accuracy effectively. Figure~\ref{fig_motivation_log} demonstrates the superiority of OptLLM in achieving multi-objective performance, compared with the results obtained by using six LLMs individually on a log parsing task (the task of parsing raw log messages into templates). Specifically, an LLM named Mixtral 8x7B~\cite{jiang2024mixtral} achieves the highest accuracy of 59.05\% at a cost of \$3.68, whereas OptLLM provides solutions that can reach a higher accuracy of 71.69\% at a lower cost of \$3.27. This highlights the potential of OptLLM to optimize the selection of LLMs for each query, resulting in improved performance and reduced cost. 

This paper considers the assignment of queries to an appropriate LLM, as a multi-objective optimization problem with the objectives of minimizing cost and maximizing performance. Evolutionary algorithms (EAs), commonly employed in optimization problems like scheduling, planning, design, and management~\cite{wu2023its,haq2022efficient}, offer a potential approach to address the query assignment for LLMs. However, EAs struggle with efficiency and may converge to suboptimal solutions, especially when tackling large-scale, intricate problems with limited available information, leading to prolonged computation times and difficulties in scaling~\cite{cheng2016test,he2020adaptive}.

To address the above challenges, this paper introduces OptLLM, an efficient and effective framework that selects a suitable LLM for a given set of queries in a predictable manner.  
OptLLM consists of two components: prediction and optimization.
The prediction component employs multi-label classification to predict the possibility of candidate LLMs processing each query successfully. To handle prediction uncertainty, OptLLM uses a weighted mean to aggregate bootstrap sample predictions and calculates the standard deviation across samples to quantify the uncertainty. The optimization component starts with two extreme solutions: one with the highest predicted accuracy and the other with the lowest cost. OptLLM then iteratively generates non-dominated solutions through a destruction and reconstruction process. The destruction phase temporarily prioritizes one objective, while the reconstruction phase optimizes the solution based on heuristic rules.
This optimization procedure enables OptLLM to address large-scale, intricate challenges with increased computational speed.

We demonstrate the generality by experimenting with various general NLP tasks and domain-specific tasks. The experimental results show that the solutions provided by OptLLM can reduce cost and improve accuracy simultaneously in all tasks. 
OptLLM achieves the same accuracy level as the best individual LLM while reducing costs by 2.40\% to 49.18\%. Compared to other multi-objective optimization algorithms, OptLLM offers solutions with accuracy improvements ranging from 2.94\% to 69.05\%, or cost savings between 8.79\% and 95.87\% while maintaining the highest accuracy attainable by each baseline algorithm. 

Our major contributions are summarized as follows: 
\begin{enumerate}
\item 
We treat the problem of assigning queries to LLMs as a multi-objective optimization problem, aiming to achieve a tradeoff between cost and accuracy. 
\item We propose OptLLM, an effective and efficient framework that automatically assigns queries to suitable LLMs. 
\item We demonstrate the generality and effectiveness of OptLLM through extensive experiments on various tasks. 
The results confirm that the solutions provided by OptLLM can reduce cost and improve accuracy across a wide range of tasks.
\end{enumerate}

\section{Background}
\label{sec2:Background}
\subsection{Large Language Model and Its Application}
In recent years, Large Language Models (LLMs) have emerged as powerful tools with remarkable capabilities in understanding and generating natural language text. 
Many LLMs, such as ChatGPT4 and Llama~\cite{openai2023gpt4, touvron2023llama}, have been developed to tackle various NLP tasks, such as text classification, question answering, sentiment analysis, and reasoning. For example, text classification involves assigning a predefined category to a given piece of text, as exemplified in Figure~\ref{fig:task_example}-(a), where the task is to classify a news headline into one of four categories: World, Sports, Business, or Sci/Tech.

\begin{figure}
    \centering
    \includegraphics[width=\linewidth]{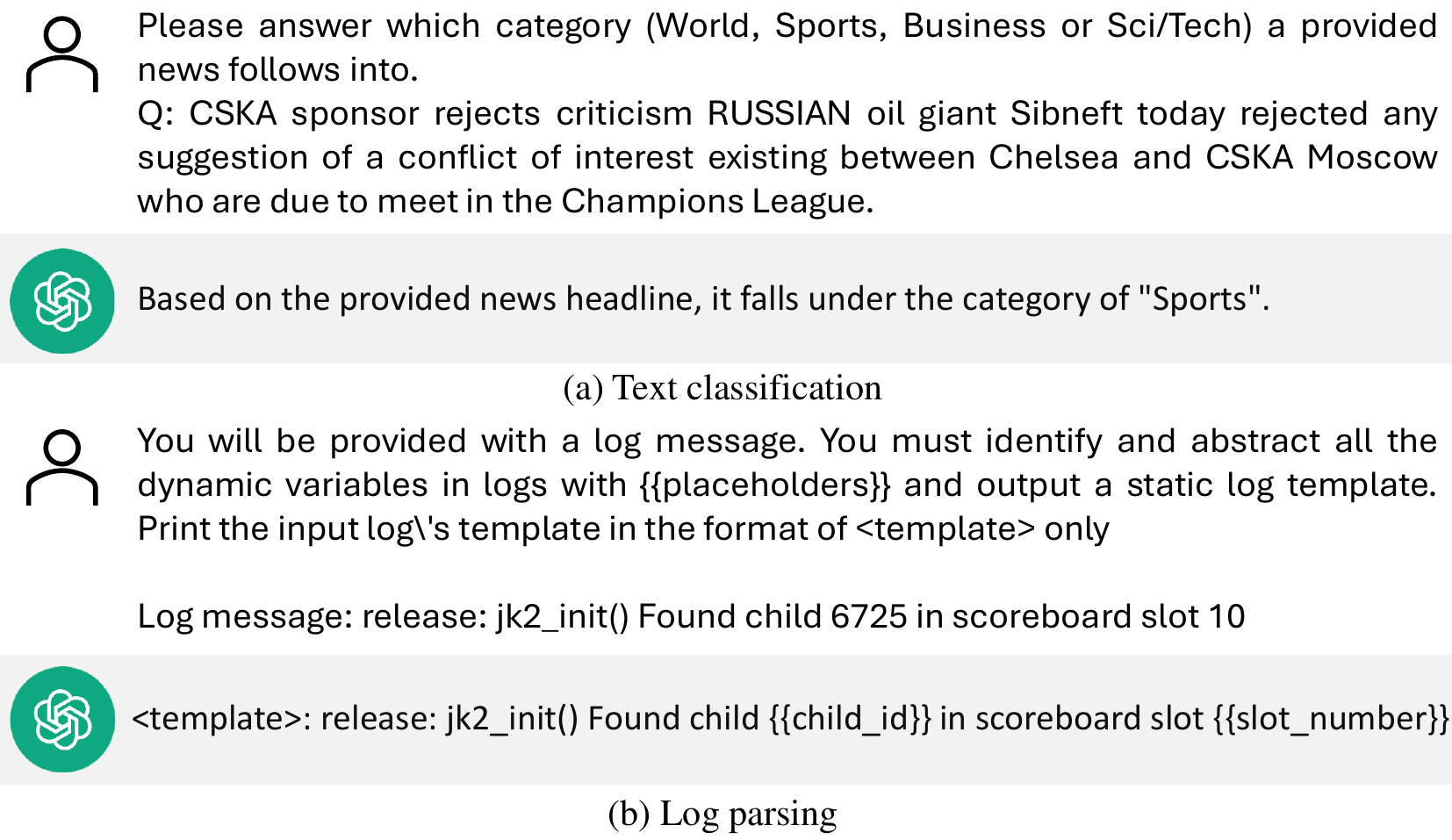}
    \caption{Examples of log parsing and text classification queries}
    \label{fig:task_example}
\end{figure}

In addition to NLP tasks, LLMs are also applied to domain-specific tasks. In this paper, we focus on LLM-based log parsing~\cite{le2023log} as a case study. Log parsing aims to extract structured information from unstructured log data generated by software systems. Log messages are system-generated lines of text containing event information. They often include dynamic variables that vary across instances of the same event type. Log parsing aims to normalize these variables by replacing them with placeholders, creating a static template that represents the common structure and content of log messages for a given event type. Figure~\ref{fig:task_example}-(b) illustrates the LLM-based log parsing process. By providing an instruction such as ``abstract variables with placeholders to extract the correct template'', we can guide the LLM to extract a template from a log message. In the example, the log message ``release: jk2\_init() Found child 6725 in scoreboard slot 10'' is parsed into the template ``release: jk2\_init() Found child \{\{child\_id\}\} in scoreboard slot \{\{slot\_number\}\}'', where the dynamic variables ``6725'' and ``10'' are replaced with the placeholders ``child\_id'' and ``slot\_number'', respectively.

\subsection{Multi-objective Optimization}
The multi-objective optimization problem (MOP) is a field of optimization concerned with problems involving multiple conflicting objectives that must be optimized simultaneously~\cite{deb2016multi}. MOP has found wide-ranging applications across various domains, including scheduling, resource allocation, and planning~\cite{ramirez2019survey}. The goal is to identify a set of Pareto-optimal solutions \cite{cheikh2010method} that represent the best possible trade-offs among the competing objectives. Formally, a multi-objective optimization problem can be formulated as:

\[ min/max\ f(x) = min/max\ (f_{1}(x), f_{2}(x), \dots, f_{M}(x))\]

where $x$ is a $\mu$-dimensional decision vector in the decision space, $f_{M}(x)$ is the $M^{th}$ objective function, and $M$ is the number of objectives.

This study employs the concept of Pareto dominance, which is the most common evaluation criterion in multi-objective optimization problems~\cite{konak2006multi}. A solution $s_{a}$ is said to dominate another solution $s_{b}$, denoted as $s_{a}\prec s_{b}$, if all objective values of $s_{a}$ are at least as good as those of $s_{b}$, and strictly better in at least one objective. Specifically, in our research context, where one objective (total cost, donated as $f_{cost}$) is to be minimized and the other (accuracy, donated as $f_{acc}$) is to be maximized. Then solution $s_{a}$ with objective values $(f_{cost}(a),f_{acc}(a))$ dominates solution $s_{b}$ with objective values $(f_{cost}(b),f_{acc}(b))$ if and only if:
\begin{itemize}
\item $f_{cost}(a) \leqslant f_{cost}(b)$ and $f_{acc}(a) > f_{acc}(b)$, or
\item $f_{cost}(a) < f_{cost}(b)$ and $f_{acc}(a) \geqslant f_{acc}(b)$
\end{itemize}

A solution $s_a$ is called a \textbf{non-dominated solution} within a set of solutions $S$ if no other solution in $S$ dominates $s_a$. A solution $s_a$ is a \textbf{Pareto-optimal solution} if no other feasible solution in the entire search space dominates $s_a$. The set of all Pareto-optimal solutions is known as the \textbf{Pareto front} or \textbf{reference point set}, representing the optimal trade-offs between the conflicting objectives. In an evolutionary algorithm, the set of all non-dominated solutions within the current population is called the \textbf{non-dominated set}. The goal of the multi-objective optimization process is to provide decision-makers with a diverse and well-balanced set of Pareto-optimal solutions, representing the best compromises among the multiple objectives.

\section{Problem Description}
\label{sec3:Problem Description}
\subsection{Problem Scope}

Numerous LLMs with different capabilities and prices are now accessible via APIs. 
OptLLM aims to identify a set of cost-effective and high-performing solutions for leveraging LLMs in query-answering tasks, encompassing both general NLP tasks (such as text classification and summarization) and domain-specific tasks (such as log parsing). As discussed in Section~\ref{sec2:Background}, multi-objective optimization problems can yield multiple Pareto-optimal solutions, unlike single-objective optimization problems typically have one solution deemed optimal. In this work, one objective is to maximize the percentage of queries processed accurately, while the other is to minimize the total cost of invoking LLM APIs.

\subsection{Problem Formulation}
Assume the user has a set of queries waiting to be processed by LLMs, represented as $J=\{j_1,j_2,\dots,j_n\}$. Each query $j_i$ is characterized by the token number $tn_i$. The user has access to a set of LLMs $L=\{l_1,l_2,\dots,l_m\}$, the price per token $price_k$ of invoking model $l_k$ is available through the LLM provider. The cost $cost_i^k$ of submitting a query $j_i$ on an LLM $l_k$ is decided by the token number of the query content and the unit price of the LLM, which is calculated as follows: 
\begin{equation}
cost_{i,k} = tn_i \times price_k 
\end{equation}

Assume $acc_{i,k}$ is the result of leveraging the LLM $l_k$ to query $j_i$, where $acc_{i,k} = 1$ represents the LLM processes the query correctly and 0 otherwise. We denote the decision variable \(x_{i,k}\), which is a binary variable that is set as 1 when query $i$ is assigned to the LLM $k$ and 0 otherwise.
This problem is then formalized as follows: 
\begin{equation}\label{obj_cost}
\text{Minimise}~f_{cost} = \sum_{i=1}^{n}cost_i
\end{equation}
\begin{equation}\label{obj_acc}
\text{Maximise}~f_{acc} = \frac{1}{n}\sum_{i=1}^{n}acc_{i}
\end{equation}
s.t.

\begin{equation}\label{job_cost}
    cost_i = \sum_{k=1}^{m}x_{i,k} \times cost_{i,k}, \forall i \in \{1, \dots, n\}
\end{equation}
\begin{equation}\label{job_acc_1}
    acc_i = \sum_{k=1}^{m}x_{i,k} \times acc_{i,k}, \forall i \in \{1, \dots, n\}
\end{equation}
\begin{equation}\label{constrain_machine}
\sum_{k=1}^{m}x_{i,k}=1, \forall i \in \{1, \dots, n\}
\end{equation}
\begin{equation}\label{constrain_x}
x_{i,k} \in \{0,1\}, \forall i \in \{1, \dots, n\}, \forall k \in \{1, \dots, m\}
\end{equation}

Equations (\ref{obj_cost}) and (\ref{obj_acc}) are the objectives to minimize total cost and improve the percentage of queries that can be processed successfully. Constraints (\ref{job_cost}) and (\ref{job_acc_1}) define the cost and accuracy for each query $i$. Constraint (\ref{constrain_machine}) is the constraint that each query should be assigned to one LLM. Constraint (\ref{constrain_x}) imposes restrictions on the decision variable.

\section{Proposed Approach}
\label{sec4OptLLM}
This section introduces OptLLM, a framework designed to efficiently allocate queries to the most suitable LLM while optimizing both cost and performance. OptLLM consists of two main components: prediction and optimization. The prediction component trains a model to estimate the probability of each LLM successfully processing a given query, while the optimization component uses the predicted probabilities to determine the optimal allocation of queries to LLMs. Figure~\ref{fig:OptLLMFramework} illustrates the architecture of OptLLM.

\begin{figure*}
    \centering
    \includegraphics[width=0.85\linewidth]
    {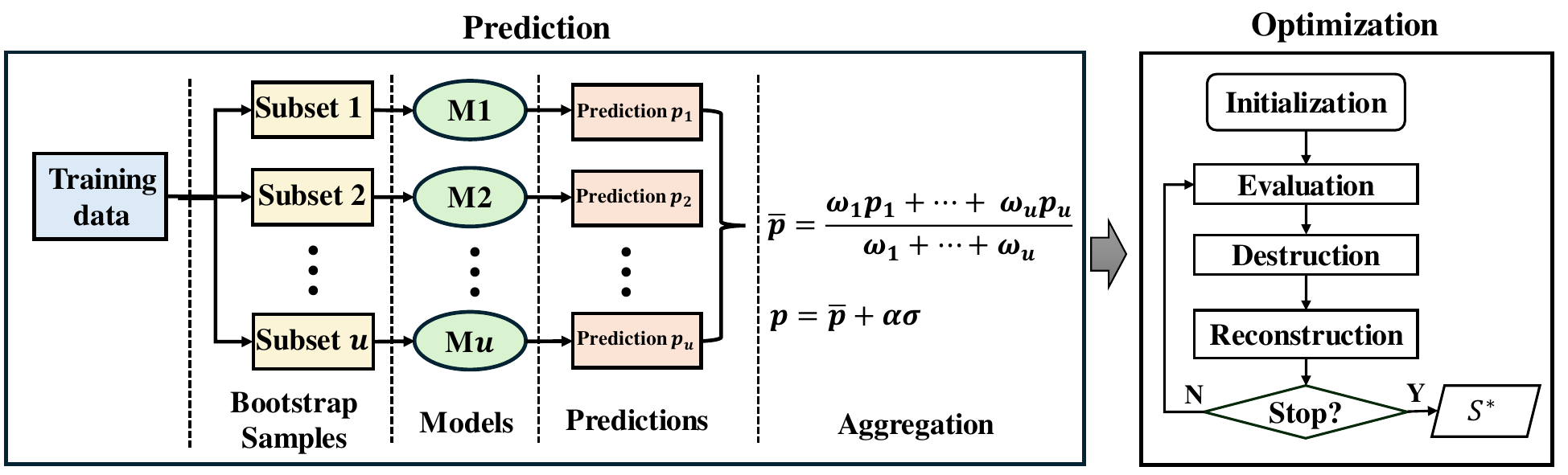}
    \caption{The Framework of OptLLM}
    \label{fig:OptLLMFramework}
\end{figure*}

\subsection{The Prediction Component}
Determining the accuracy of LLMs on a specific input is challenging, as it remains uncertain until that input is actually processed by the LLM. To address this issue, we propose a multi-label classification model that incorporates uncertainty estimation techniques to enhance the reliability and stability of the performance predictions.

\subsubsection{Bootstrap-based training and aggregation}
The multi-label classifier takes the query content as input and predicts the accuracy of each candidate LLM for that specific query.
The training set is collected by querying a small set of queries with each candidate LLM and recording their responses. Herein, a pre-trained word embedding model extracts features from the input query content, whilst the response results of the queried LLM are considered as labels.
The word embeddings, which are trained on a large corpus of text data, enable the resulting feature vectors to capture more nuanced and contextual information about the input compared to traditional bag-of-words or count-based representations. Subsequently, we construct a set of prediction models trained with an ensemble of Random Forest classifiers on multiple bootstrap samples of the training and validation data. To obtain the final predictions, we aggregate the predictions from each bootstrap sample using a weighted mean approach. 
Assume there are $u$ bootstrap samples, and for each sample, we evaluate the model's performance on the validation data by calculating the percentage of correctly predicted data, denoted as $\omega$. 
The weighted mean prediction $\bar{p}$ is calculated as follows:
\begin{equation}\label{mean_p}
\bar{p} = \frac{\omega_1 p_1 + \dots \omega_u p_u}{\omega_1 + \dots \omega_u}
\end{equation}

\subsubsection{Robust-aware predicted accuracy in aggregation}
Predictions play a crucial role in the selection process during optimization. It is important to note that when the prediction model is not highly accurate, there is a risk that OptLLM may choose a solution with higher predicted accuracy but worse real-world performance. To mitigate this risk, we employ a robust optimization in aggregation phrase, a common approach for handling uncertainty by leveraging interval or discrete data from statistics or historical observations~\cite{ben2009robust}. Specifically, we calculate the standard deviation $\sigma$ of the predictions across the bootstrap samples to quantify the uncertainty. The robust-aware predicted accuracy is defined as follows. Formally, let $p_{i,k}$ represent the predicted accuracy of query $j_i$ on LLM $l_k$. We formulate $p_{i,k}$ as:
\begin{equation}\label{robust_p}
p_{i,k} = \bar{p}_{i,k} + \alpha_{i,k} \times \sigma_{i,k}
\end{equation}
where $\bar{p}_{i,k}$ is the weighted mean prediction, and $\alpha_{i,k}$ is a parameter that controls the robustness of the predicted accuracy by adjusting the uncertainty term $\sigma_{i,k}$.

By incorporating uncertainty information into the predictions, our model can adapt its outputs based on the level of uncertainty, providing more robust and reliable results. The $\alpha$ parameter determines the robustness of the predictions by governing the impact of the uncertainty term on the final predictions. When $\alpha$ is set to 0, the uncertainty term is disregarded, and the robust predictions are equivalent to the mean predictions, indicating no additional consideration of uncertainty. Positive values of $\alpha$ increase the robustness by adding the uncertainty term to the mean predictions, with larger values resulting in more conservative predictions that account for a wider range of possible accuracies. Conversely, negative values of $\alpha$ decrease the robustness by subtracting the uncertainty term from the mean predictions, with larger negative values yielding more aggressive predictions that consider a narrower range of possible accuracies. During the search process, the accuracy objective of a candidate allocation solution is determined as follows:

\begin{equation}\label{pred_acc}
f_{acc}^{'} = \frac{1}{n} \sum_{i=1}^{n} p_i
\end{equation}

\begin{equation}\label{job_predicted_acc}
p_i = \sum_{k=1}^{m} x_{i,k} \times p_{i,k}, \quad \forall i \in \{1, \dots, n\}
\end{equation}

where $f_{acc}^{'}$ represents the predicted overall accuracy objective of the solution, and $p_i$ is the predicted accuracy for query $j_i$ based on the LLM assignments in the current solution.

\subsection{The Optimization Component}
As shown in Algorithm~\ref{alg:architecture_moo}, the optimization component consists of initialization and heuristic search phases. In the initialization phase, OptLLM generates two extreme optimal solutions $s_h$ and $s_c$ (Line 2 in Algorithm~\ref{alg:architecture_moo}), where $s_h$ has the highest predicted accuracy and $s_c$ has the lowest cost. Then, the search phase of OptLLM performs destruction and reconstruction iteratively (Lines 7-13 in Algorithm~\ref{alg:architecture_moo}) to generate a set of non-dominated solutions.

\begin{algorithm}
\small
\caption{Function \textit{Optimization}}
\label{alg:architecture_moo}
\begin{algorithmic}[1]
\Require
$P$: predicted accuracy table;
$C$: cost table;
$GN$: grid parameter
\Ensure $S^*$: non-dominated solutions output by the OptLLM;
\State $S^* \leftarrow \emptyset$ 
\State $s_{h}, s_{c} \gets$ Initialization($P, C$) 
\State $\Delta{cost} \gets \frac{f_{cost}(s_{h})-f_{cost}(s_{c})}{GN}$
\State $\Delta{acc} \gets \frac{f_{acc}(s_{h})-f_{acc}(s_{c})}{GN}$
\State $s_1$ $\leftarrow$ $s_{h}$
\State $s_2$ $\leftarrow$ $s_{c}$
\While{no termination criterion is met}
        \State $\bar{s}_1, \bar{s}_2 \gets \Call{destruction\_strategy}{(s_1, \Delta{cost}), (s_2, \Delta{ccc})}$
        \State $\hat{s}_1, \hat{s}_2 \gets \Call{reconstruction\_strategy}{\bar{s}_1, \bar{s}_2}$
        \State $S^* \gets \hat{s}_1, \hat{s}_2$
        \State $s_1 \gets \hat{s}_1$
        \State $s_2 \gets \hat{s}_2$
\EndWhile
\State \Return $S^*$
\end{algorithmic}
\end{algorithm}

\subsubsection{Initialization}
In the initialization phase, we generate two optimal solutions according to the concept of Pareto dominance. Let $s_{c}$ denote the solution obtained by assigning each query to the cheapest LLM, resulting in the lowest cost. On the other hand, $s_{h}$ represents the solution achieved by selecting the LLM with the highest predicted accuracy for each query. To prove that both $s_{c}$ and $s_{h}$ are Pareto solutions, we introduce the following lemma:
\begin{lemma} \label{lemma1}
If $s_{c}$ represents the solution with the lowest cost and $s_{h}$ represents the solution with the highest accuracy, then both $s_{c}$ and $s_{h}$ are Pareto solutions.
\end{lemma}
\begin{proof}
If $s_{c}$ is not a Pareto solution, then there must exist a solution $s$ such that $s$ dominates $s_{c}$ in at least one objective while being no worse in the other objective. Formally, $s \prec s_{c}$, where ``$\prec$'' indicates that solution $s$ dominates solution $s_{c}$.

Let ${f}_{cost}$ denote the cost objective and ${f}_{acc}$ denote the accuracy objective. If $s \prec s_{c}$, then ${f}_{cost}(s) < {f}_{cost}(s_{c})$ and ${f}_{acc}(s) \geq {f}_{acc}(s_{c})$, or ${f}_{cost}(s) \leq {f}_{cost}(s_{c})$ and ${f}_{acc}(s) > {f}_{acc}(s_{c})$. However, $s_{c}$ is the solution with the lowest cost, where ${f}_{cost}(s_{c}) < {f}_{cost}(s)$. This contradicts the definition of dominance, as $s_{c}$ should be at least as good as $s$ in all objectives and strictly better in at least one objective.

Similarly, $s_{h}$ represents the solution with the highest accuracy, where ${f}_{acc}(s_{h}) > {f}_{acc}(s)$ for any solution $s$. According to the definition of Pareto dominance, there is no solution dominating $s_{c}$ or $s_{h}$. Therefore, both $s_{c}$ and $s_{h}$ are Pareto solutions. This completes the proof.
\end{proof}

\subsubsection{The heuristic search}
During the heuristic search phase, OptLLM employs a process of destruction and reconstruction to iteratively generate a set of non-dominated solutions. During the destruction stage, OptLLM removes elements (queries have been assigned) from the current solution and reassigns them to maximize one objective while temporarily disregarding the other. This approach allows for a focused optimization of the prioritized objective. There are two search directions. Assume $s_h$ is the solution with the highest expected accuracy with a relatively high associated cost. In this case, OptLLM first destructs $s_h$ to release as much cost as possible. Conversely, $s_c$ represents the cheapest solution with a lower accuracy. During the destruction, OptLLM prioritizes increasing accuracy to the greatest extent possible by temporarily disregarding the cost objective. Subsequently, the reconstruction stage refines the solution obtained from destruction, transforming it into a non-dominated solution using a predefined scoring function that balances both objectives. For instance, as illustrated in Figure \ref{fig:assignment_multi}, solution $B$ is a feasible solution derived from the non-dominated solution $A$ by releasing sufficient cost. Subsequently, solution $C$ is obtained by reconstructing solution $B$, resulting in a new non-dominated solution.

\begin{figure}
    \centering
    \includegraphics[width=1\linewidth]{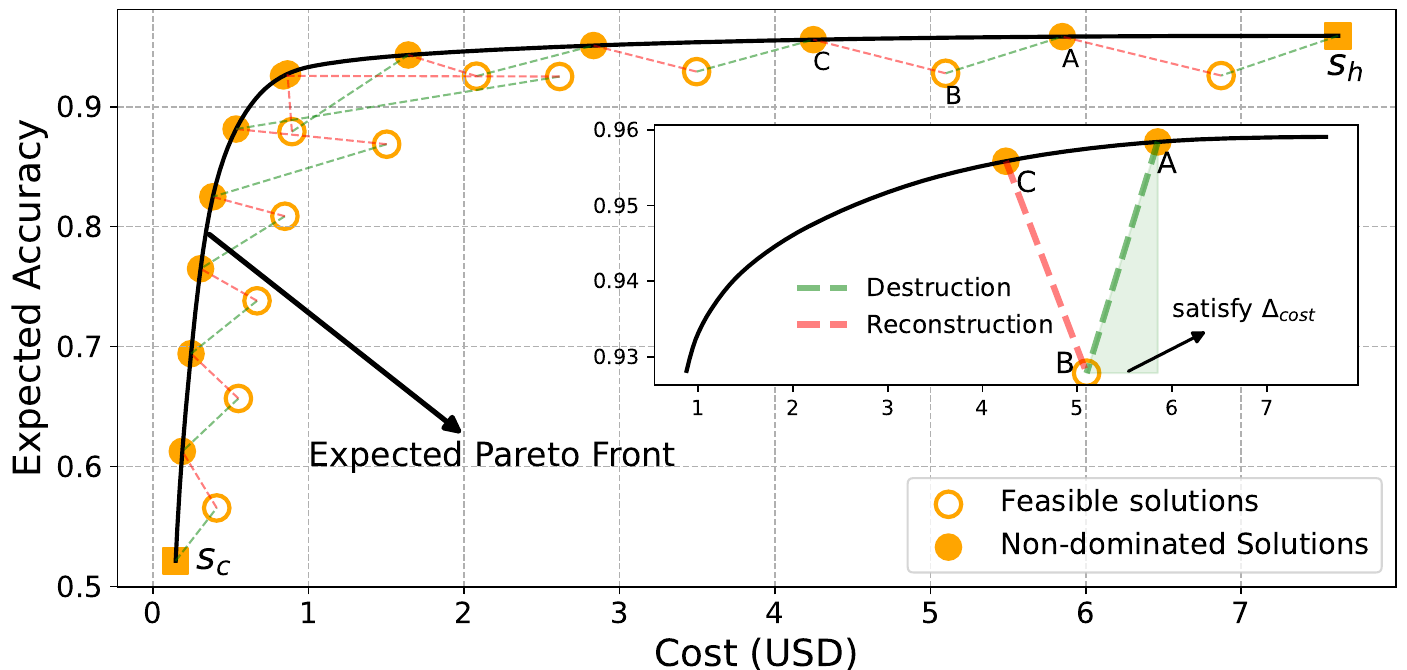}
    \caption{Solutions by destruction and reconstruction in optimization component. $B$ is a feasible solution derived from the non-dominated solution $A$ by releasing sufficient cost, and $C$ is a new non-dominated solution obtained by reconstructing $B$.}
    \label{fig:assignment_multi}
\end{figure}

\noindent\textbf{Destruction.}
Let $\Hat{s}$ denote the solution obtained after the destruction phase. Taking optimizing cost objective as an example, $\Hat{s}$ is expected to achieve cost savings of $gap_{cost}$ compared to the original solution $s$, expressed as $f_{cost}(\Hat{s})-f_{cost}(s)\geq gap_{cost}$ (Lines 2-6 in Algorithm~\ref{alg:destruction}). The algorithm iteratively reassigns queries, beginning with those offering the largest cost reductions, until the cost of the solution $\Hat{s}$ meets the specified condition. For a query $j_i$, if its current assignment is LLM $k'$ and will be reassigned to $k$, the cost difference $cs_{i,k}$ is calculated as follows:
\[ cs_{i,k} = cost_{i,k} - cost_{i,k'}, k\neq k'
 \]
For each query $j_i$, $cs_i$ represents the largest cost savings achievable by selecting another LLM, calculated as:
\[ cs_i = \min \{ cs_{i,1}, cs_{i,2}, \dots, cs_{i,m}\} \]

Similarly, if the goal is to optimize accuracy through destruction, $\Hat{s}$ is expected to achieve an accuracy improvement of $gap_{acc}$ compared to the original solution $s$ (Lines 8-12 in Algorithm~\ref{alg:destruction}). This is expressed as $f_{acc}(\Hat{s})-f_{acc}(s)\geq gap_{acc}$. The accuracy improvement ($ai_{i,k}$) associated with reallocating query $j_i$ from LLM $k'$ to LLM $k$ is defined as follows:
\[ ai_{i,k} = acc_{i,k} - acc_{i,k'}, k\neq k' \]
Let $ai_i$ denote the maximum accuracy improvement achievable by selecting an alternative LLM for query $j_i$, calculated as:
\[ ai_i = \max \{ ai_{i,1}, ai_{i,2}, \dots, ai_{i,m}\} \]

\begin{algorithm}
\small
\caption{Function \textit{Destruction}}
\label{alg:destruction}
\begin{algorithmic}[1]
\Require
$gap_{cost}$ (or $gap_{acc}$): cost (or accuracy) gap
\Ensure $\Hat{s}$: solution output by the function \textit{Destruction}
\If{$gap_{cost}$}
\While {$f_{cost}-\Hat{f}_{cost}<gap_{cost}$} 
\State $i=Min(cs_i)$ \Comment{query with the largest cost reduction}
\State $\Hat{s} \leftarrow$ reallocation($s, i$)
\State $s \gets \Hat{s}$
\EndWhile
\ElsIf{$gap_{acc}$}
\While {$\Hat{f}_{acc}-f_{acc}<gap_{acc}$} 
\State $i=Max(ai_i)$ \Comment{query with the largest accuracy improvement}
\State $\Hat{s} \leftarrow$ reallocation($s, i$)
\State $s \gets \Hat{s}$
\EndWhile
\EndIf
\State \Return $\Hat{s}$
\end{algorithmic}
\end{algorithm}

\begin{algorithm}
\small
\caption{Function \textit{Reconstruction}}
\label{alg:reconstruction}
\begin{algorithmic}[1]
\Require $\Hat{s}$: solution output by the function \textit{Destruction}
\Ensure $\Bar{s}$: solution output by the function \textit{Reconstruction}
 
\State $i_1 \leftarrow$ query with the largest positive score $score_{i_1}$ in $s$
\State $i_2 \leftarrow$ query with the largest negative score $score_{i_2}$ in $s$
\While{$|score_{i_1}| > |score_{i_2}|$ and $score_{i_1} \cdot score_{i_2} < 0$}
    \State $\hat{s} \leftarrow$ reallocation($s, i_1, i_2$)
    \If{$(s \nsucc \hat{s}) \land (\hat{s} \nsucc s)$} \Comment{\textit{Neither $\hat{s}$ nor $s$ dominates each other}} 
        \State $\hat{S} \leftarrow \hat{S} \cup \{\hat{s}\}$
    \ElsIf{$\hat{s} \prec s$}
        \State $\hat{S}\leftarrow \{\hat{s}\} \cup (\hat{S} \setminus \{s\})$ \Comment{\textit{replace $s$ with $\hat{s}$}}     
    \EndIf
    \State $s \leftarrow \hat{s}$
    \State update $i_1, i_2, score_{i_1}, score_{i_2}$
\EndWhile
\State \Return $\hat{S}$
\end{algorithmic}
\end{algorithm}
\noindent\textbf{Reconstruction.}
Then, we perform reconstruction to optimize solution \(\hat{s}\). During this process, the objective is to maximize the accuracy improvement per unit cost. In other words, for a given solution, a higher ratio of accuracy improvement to cost increase is considered more desirable. Based on this thought, we introduce the following lemma.

\begin{lemma} \label{lemma2}
Assume there are two solutions \(s_1\) and \(s_2\). If \(\frac{f_{acc}(s_1)}{f_{cost}(s_1)} \leq \frac{f_{acc}(s_2)}{f_{cost}(s_2)}\), then \(s_1 \nprec s_2 \).
\end{lemma}

\begin{proof}
If \(s_1 \prec s_2\), based on the definition, there are two possible situations:
1) \(f_{cost}(s_1) \leq f_{cost}(s_2)\) and \(f_{acc}(s_1) > f_{acc}(s_2)\).
2) \(f_{cost}(s_1) < f_{cost}(s_2)\) and \(f_{acc}(s_1) \geq f_{acc}(s_2)\). 

In the first situation, if \(f_{acc}(s_1) > f_{acc}(s_2)\), then \(\frac{f_{acc}(s_1)}{f_{cost}(s_2)} > \frac{f_{acc}(s_2)}{f_{cost}(s_2)}\). Because \(f_{cost}(s_1) \leq f_{cost}(s_2)\), we can easily have \(\frac{f_{acc}(s_1)}{f_{cost}(s_1)} \geq \frac{f_{acc}(s_1)}{f_{cost}(s_2)} > \frac{f_{acc}(s_2)}{f_{cost}(s_2)}\). Similarly, in the second situation, we can conclude that \(\frac{f_{acc}(s_1)}{f_{cost}(s_1)} > \frac{f_{acc}(s_1)}{f_{cost}(s_2)} \geq \frac{f_{acc}(s_2)}{f_{cost}(s_2)}\).

Therefore, we can have if \(s_1 \prec s_2 \), then \(\frac{f_{acc}(s_1)}{f_{cost}(s_1)} > \frac{f_{acc}(s_2)}{f_{cost}(s_2)}\). This is logically equivalent to: if \(\frac{f_{acc}(s_1)}{f_{cost}(s_1)} \leq \frac{f_{acc}(s_2)}{f_{cost}(s_2)}\), then \(s_2 \nprec s_1 \). This completes the proof.
\end{proof}

Based on Lemma~\ref{lemma2}, we reassign the queries in \(\hat{s}\) until the quotient of accuracy and cost (i.e., \(\frac{f_{acc}}{f_{cost}}\)) cannot be increased further. Let $s$ be the original solution and $s'$ be the solution after the reassignment. The difference between $s$ and $s'$ is caused by the reassignment of a query or a set of queries. By iteratively reassigning queries and applying the criterion from Lemma~\ref{lemma2}, we arrive at a non-dominated solution $\bar{s}$ that cannot be further improved, indicating that there is no new solution possible that can dominate the current solution.

\section{Experimental Design}
\label{sec5:Experimental Design}
\subsection{Research Questions}
We aim to answer the following research questions (RQs) through experiments.

\noindent RQ1: Can OptLLM allocate queries to LLMs effectively and efficiently?

\noindent RQ2: Does each core component of OptLLM contribute to the overall performance?

\noindent RQ3: How do different hyper-parameter settings affect the performance of OptLLM?

\begin{table}[ht]
\caption{Benchmark specification}
\label{tb_NLPjobs}
\centering
\resizebox{\columnwidth}{!}{%
\begin{tabular}{ccccc}
\hline
Dataset & Task Type & \begin{tabular}[c]{@{}c@{}}Train size \\ (1\%)\end{tabular} & \begin{tabular}[c]{@{}c@{}}Val size\\ (1\%)\end{tabular} & \begin{tabular}[c]{@{}c@{}}Test size \\ (98\%)\end{tabular} \\ \hline
LogPai & Log parsing & 320 & 320 & 31360 \\
AGNEWS & Text classification & 76 & 76 & 7448 \\
COQA & Question answering & 80 & 80 & 7822 \\
HEADLINES & Sentiment analysis & 100 & 100 & 9800 \\
SCIQ & Reasoning & 117 & 117 & 11443 \\ \hline
\end{tabular}%
}
\end{table}
\subsection{Benchmark Setting}
\subsubsection{Natural Language Processing (NLP) tasks}
To show the generality of OptLLM on different types of tasks, we have chosen four NLP tasks, including text classification(AGNEWS~\cite{zhang2015character}), question answering (COQA~\cite{reddy2019coqa}), sentiment analysis (HEADLINES~\cite{sinha2021impact}),  and reasoning( SCIQ~\cite{welbl2017crowdsourcing}). 12 candidate LLMs are selected from 4 mainstream providers: OpenAI\textsuperscript{\ref{openai}} (GPT-Curie, ChatGPT, GPT-3, and GPT-4), AI21\textsuperscript{\ref{ai21}} (Jurassic-1 Large, Jurassic-1 Grande, and Jurassic-1 Jumbo), Cohere\footnote{https://cohere.com/} (Xlarge and Medium), and Textsynth~\footnote{https://textsynth.com/} (GPT-J, FAIRSEQ, and GPT-Neox).
The raw data is provided by Chen et al. ~\cite{chen2023frugalgpt}, which contains the inputs (prompts) sent to the LLMs, ground truth references, LLM outputs, and cost. 

\subsubsection{Domain-specific tasks}
Furthermore, we have chosen an intelligent software engineering (SE) task, specifically focusing on LLM-based log parsing. We utilize log data sourced from the LogPai benchmark~\cite{zhu2019tools, khan2022guidelines} to interface with 8 LLMs, including TogertherAI\textsuperscript{\ref{together}} (Mixtral-8x7B, Llama-2-7B, Llama-2-13B, Llama-2-70B, Yi-34B, Yi-6B), AI21\textsuperscript{\ref{ai21}}(Jurassic-2 Mid and Jurassic-2 Ultra). 
The LogPai benchmark consists of log data from 16 systems, including distributed systems, supercomputers, operating systems, mobile systems, server applications, and standalone software. The raw data includes inputs (queries and full prompts) sent to the LLMs, ground truth references, LLM outputs, and the corresponding execution cost. 
The details of the datasets are listed in Table~\ref{tb_NLPjobs}.

\subsection{Baselines}

\subsubsection{Individual LLM} Because it is common practice to assign a task to a specific LLM, we select individual LLMs as baselines to compare against OptLLM's performance. 
We submit the entire set of queries to each individual LLM and assess the resulting cost and the proportion of queries that have been successfully completed. 

\subsubsection{Classic multi-objective optimization algorithms}
As shown in Algorithm~\ref{alg:architecture_moo}, OptLLM utilizes a heuristic search-based algorithm in optimization. We compare the effectiveness of this algorithm with well-known multi-objective optimization algorithms, including the Non-dominated Sorting Genetic Algorithm (NSGA-\rom{2})\cite{deb2002fast}, Multi-objective Particle Swarm Optimisation (MOPSO)\cite{coello2002mopso}, and Multi-objective Evolutionary Algorithm with Decomposition (MOEA/D)\cite{zhang2007moea}. These three algorithms have been extensively studied and have proven to be effective in solving a wide range of multi-objective optimization problems~\cite{ramirez2019survey, chen2018beyond}. In addition, three variants of classic algorithms are also compared, including R-NSGA-\rom{2}\cite{deb2006reference}, SMS-EMOA~\cite{beume2007sms}, and MOEA/D-GEN~\cite{wang2018scalable}. It is important to note that all the evaluated multi-objective optimization algorithms are integrated with the same prediction component as OptLLM, to enable a fair comparison of the optimization strategies. 

\subsection{Evaluation Metrics}
\subsubsection{Evaluating single solution performance}
When evaluating the performance of a single solution, such as assigning all queries to a single LLM, directly comparing the optimization objectives is feasible. This approach allows for a straightforward assessment of the solution's effectiveness in terms of cost and accuracy. Two objectives considered in this evaluation are:
\begin{itemize}
    \item $f_{cost}$: the total cost of invoking LLM APIs (see Equation~\ref{obj_cost})
    \item $f_{acc}$: the percentage of jobs that have been processed correctly (see Equation~\ref{obj_acc})
\end{itemize}

\subsubsection{Multi-objective optimization evaluation metrics}
To evaluate the performance of multi-objective optimization algorithms, we introduce widely used metrics: IGD and $\Delta$~\cite{audet2021performance}. The metrics require knowledge of the Pareto front, which represents the set of optimal trade-off solutions. We obtain the Pareto front through an exhaustive enumeration search, which serves as a reference for calculating the evaluation metrics.

\textbf{True Pareto Front Generation:}
Based on the definition of Pareto dominance, an extremely optimal solution $s_{cheapest}$ can be easily obtained, where all jobs are allocated to the cheapest LLM to achieve the lowest cost. Then, we generate a new solution that is not dominated by $s_{cheapest}$. Specifically, we calculate the cost required by a new LLM to process the job correctly for all the jobs that are not solved in the current solution. The job with the smallest cost is reallocated to get a new optimal solution. Repeat this process until the accuracy cannot be increased.

\textbf{Inverted Generational Distance (IGD)}: The IGD measures the distance between the obtained solution set and the true Pareto front (or the reference set), evaluating the quality of the obtained solution set~\cite{audet2021performance}. Let the solution set $\Lambda=\{y_1, y_2, \dots, y_{|\Lambda|}\}$ be the Pareto front for the problem and $\hat{\Lambda}=\{\hat{y}_1, \hat{y}_2, \dots, \hat{y}_{|\hat{\Lambda}|}\}$ be the solutions (approximation) obtained by an algorithm. We can define IGD by:
\[
IGD(\hat{\Lambda}, \Lambda) = \frac{1}{|\Lambda|} \sqrt{\sum_{y \in \Lambda} \left(\min_{\hat{y} \in \hat{\Lambda}} d(\hat{y}, y)\right)^2}
\]
where \(|\Lambda|\) denotes the number of solutions in the true Pareto front (or the reference set). \(d(\hat{y}, y)\) is the Euclidean distance (or any other appropriate distance metric) between a solution \(\hat{y}\) in the obtained solution set \(\hat{\Lambda}\) and the nearest solution \(y\) in the true Pareto front \(\Lambda\). A lower value of IGD means a better performance, indicating the obtained solution set has a better distribution and better approximation to the reference set. 

\textbf{$\Delta$ metric}: The $\Delta$ metric assesses the diversity and uniformity of the solutions distribution along the Pareto front by measuring Euclidean distances between consecutive solutions and comparing them to the average distance~\cite{audet2021performance}.
\[
\Delta(\hat{\Lambda}) = \frac{d_{f} + d_{l} + \sum_{z=1}^{|\hat{\Lambda}|-1} |d_{z} - \Bar{d}|}{d_{f} + d_{l} + (|\hat{\Lambda}|-1)\Bar{d}}
\]

where $d_{f}$ and $d_{l}$ are the Euclidean distances between the extreme solutions and their nearest neighbors in the obtained solution set $\hat{\Lambda}$. 
$d_z$ denotes the Euclidean distance between the $z$th solution $\hat{y}_z$ and its next solution $\hat{y}_{z+1}$ in the obtained solution set $\hat{\Lambda}$, sorted in ascending order of a chosen objective. $\Bar{d}$ represents the mean value of the distances $d_z$. A smaller $\Delta$ metric indicates a higher diversity and distribution of solutions.

To ensure a fair comparison, all multi-objective algorithms have the same termination condition of 200 iterations. We also record the execution time (in minutes) required by each algorithm to generate the solution sets within this fixed number of iterations to evaluate the computational efficiency.

\subsection{Implementation}
To mitigate the impact of randomness and ensure the robustness of our findings, each experiment is conducted 10 times and the average value is reported. The implementations of all algorithms are based on Python 3.11. Specifically, we utilize the standard versions of NSGA-\rom{2}, R-NSGA-\rom{2}, and SMS-EMOA from the Pymoo library~\cite{pymoo}, while MOPSO and MOEA/D are obtained from the Pygmo library.

\textbf{Accuracy table:} The training and validation sets each comprise 1\% of the data, while the remaining data is used for testing. All multi-objective optimization algorithms use the same predicted accuracy generated by the prediction component for evaluation during the search process.

\textbf{Parameter setting:} Optuna, a widely used hyperparameter optimization package~\cite{akiba2019optuna}, is employed to ensure the effectiveness and efficiency of all algorithms. We conduct parameter tuning using Optuna to choose optimal parameter settings for all algorithms. For OptLLM, bootstrap sample number $\mu$ is set as 100, $GN$ is set as 50, and $\alpha$ is set as 0.5. Due to space constraints, the results and optimal parameters for all baselines are provided on our project webpage\footnote{\label{github}https://github.com/superyue72/OptLLM}.

\section{Experimental Results}
\label{sec6:Experimental Results}
\subsection{RQ1: Comparison with the Baselines}
\subsubsection{Comparison with individual LLMs}
OptLLM generates a set of solutions rather than a single solution. For comparison purposes, we select the solution from OptLLM that achieves the same accuracy as the best-performing individual LLM for each dataset. Table~\ref{tb_comparison_single} presents the cost savings achieved by OptLLM compared to the best individual LLM. For instance, on the AGNEWS dataset, GPT-4 outperforms other individual LLMs, attaining an accuracy of 0.90 at a cost of 126.58. OptLLM provides a solution with the same accuracy but at a lower cost of 75.77, resulting in a 40.14\% cost saving. Across all datasets, OptLLM maintains the same level of accuracy as the best individual LLM while significantly reducing costs. The cost savings range from 2.40\% for the SCIQ dataset to 49.18\% for the LogPai dataset, demonstrating OptLLM's ability to achieve comparable performance to the best individual LLM while offering cost savings.

\begin{center}
\vspace{-12pt}
\begin{table}[ht]
\caption{Cost savings by OptLLM compared with the individual LLM ($f_{cost}$ and $f_{acc}$ are the cost and accuracy of the solution )}
\label{tb_comparison_single}
\centering
\resizebox{\columnwidth}{!}{%
\begin{tabular}{|c|ccc|cc|}
\hline
\multirow{2}{*}{Dataset} & \multicolumn{3}{c|}{Best Individual LLM} & \multicolumn{2}{c|}{Cost to reach the same accuracy} \\ \cline{2-6} 
 & Model & $f_{acc}$ & $f_{cost}$ & $f_{cost}$ & Cost Savings \\ \hline
AGNEWS & GPT-4 & 0.90 & 126.58 & 75.77 & 40.14\% \\
COQA & GPT-4 & 0.27 & 216.01 & 152.63 & 29.34\% \\
HEADLINES & GPT-4 & 0.86 & 65.28 & 40.91 & 37.33\% \\
SCIQ & Xlarge & 0.71 & 144.86 & 141.39 & 2.40\% \\
LogPai & Mixtral-8x7B & 0.59 & 3.68 & 1.87 & 49.18\% \\ \hline
\end{tabular}%
}
\end{table}
\vspace{-18pt}
\end{center}

\subsubsection{Comparison with classic multi-objective optimization algorithms}
Table~\ref{tb_comparison_perf} presents the best solution with the highest accuracy generated by all algorithms. The results demonstrate that OptLLM consistently provides solutions with the highest accuracy among all the compared methods. For instance, OptLLM can generate a solution reaching an accuracy of 0.71, while the baseline algorithms can only provide solutions with an accuracy of around 0.45. Furthermore, Table~\ref{tb_comparison_saving} illustrates the cost reduction achieved by OptLLM while maintaining the same accuracy as the highest accuracy attained by the baseline methods. Take the AGNEWS dataset as an example, MOPSO, the second-best performer among all the algorithms, achieves an accuracy of 82.18\% at a cost of 51.15 dollars. While OptLLM attains the same accuracy for only 5.1 dollars, resulting in a substantial cost saving of 90.03\%. This significant cost reduction demonstrates OptLLM's ability to provide high-performing solutions under strict cost constraints, making it a more cost-effective option compared to the baseline methods.

Table~\ref{tb_comparison_metric} presents the IGD and $\Delta$ of the solutions generated by OptLLM and other multi-objective optimization methods on the benchmarks. OptLLM achieves the best performance on all benchmark instances. The IGD value of the solutions found by OptLLM is much smaller than that achieved by its competitors. For example, on the AGNEWS dataset, OptLLM can achieve an IGD of 0.13, while the second best value is 11.61 by MOPSO. This indicates that the solutions generated by OptLLM are much closer to the Pareto Front, with a lower cost and higher accuracy. In terms of $\Delta$, OptLLM can obtain the smallest value in all scenarios, showing the best diversity of the obtained solution set. In terms of computation time, to generate the same number of solutions, OptLLM is around 5 times faster than its competitors. 

\begin{center}
\begin{table*}[]
\caption{The solution with the highest accuracy by all algorithms}
\label{tb_comparison_perf}
\center
{%
\centering
\begin{tabular}{|c|cc|cc|cc|cc|cc|}
\hline
Dataset & \multicolumn{2}{c|}{AGNEWS} & \multicolumn{2}{c|}{COQA} & \multicolumn{2}{c|}{HEADLINES} & \multicolumn{2}{c|}{SCIQ} & \multicolumn{2}{c|}{LogPai} \\ \hline
Alg & $f_{acc}$ & IMPROV & $f_{acc}$ & IMPROV & $f_{acc}$ & IMPROV & $f_{acc}$ & IMPROV & $f_{acc}$ & IMPROV \\ \hline
OptLLM & \textbf{0.90} & \textbackslash{} & \textbf{0.27} & \textbackslash{} & \textbf{0.86} & \textbackslash{} & \textbf{0.70} & \textbackslash{} & \textbf{0.71} & \textbackslash{} \\
NSGA-II & 0.74 & 21.62\% & 0.21 & 28.57\% & 0.82 & 4.88\% & 0.68 & 2.94\% & 0.42 & 69.05\% \\
MOPSO & 0.82 & 9.76\% & 0.23 & 17.39\% & 0.80 & 7.50\% & 0.68 & 2.94\% & 0.48 & 47.92\% \\
MOEA/D & 0.70 & 28.57\% & 0.22 & 22.73\% & 0.78 & 10.26\% & 0.67 & 4.48\% & 0.46 & 54.35\% \\
RNSGA-II & 0.77 & 16.88\% & 0.22 & 22.73\% & 0.79 & 8.86\% & 0.67 & 4.48\% & 0.46 & 54.35\% \\
SMS-EMOA & 0.76 & 18.42\% & 0.22 & 22.73\% & 0.79 & 8.86\% & 0.67 & 4.48\% & 0.45 & 57.78\% \\
MOEA/D-GEN & 0.82 & 9.76\% & 0.22 & 22.73\% & 0.80 & 7.50\% & 0.68 & 2.94\% & 0.47 & 51.06\% \\ \hline
\end{tabular}%
}
\end{table*}

\end{center}

\begin{center}
\begin{table*}[ht]
\caption{Cost ($f_{cost}$) savings by OptLLM to match the baseline's performance}
\label{tb_comparison_saving}
\resizebox{2\columnwidth}{!}{%
\centering
\begin{tabular}{|c|ccc|ccc|ccc|ccc|ccc|}
\hline
Dataset & \multicolumn{3}{c|}{AGNEWS} & \multicolumn{3}{c|}{COQA} & \multicolumn{3}{c|}{HEADLINES} & \multicolumn{3}{c|}{SCIQ} & \multicolumn{3}{c|}{LogPai} \\ \hline
Alg & Baselines & OptLLM & Savings & Baselines & OptLLM & Savings & Baselines & OptLLM & Savings & Baselines & OptLLM & Savings & Baselines & OptLLM & Savings \\ \hline
NSGA-II & 30.84 & 2.35 & 92.38\% & 26.63 & 16.91 & 36.50\% & 4.78 & 3.40 & 28.87\% & 28.70 & 20.14 & 29.83\% & 11.93 & 1.83 & 84.66\% \\
MOPSO & 51.15 & 5.10 & 90.03\% & 36.17 & 32.99 & 8.79\% & 14.63 & 1.96 & 86.60\% & 80.30 & 41.11 & 48.80\% & 3.10 & 2.07 & 33.23\% \\
MOEA/D & 16.21 & 1.54 & 90.50\% & 25.74 & 17.40 & 32.40\% & 6.67 & 0.62 & 90.70\% & 44.78 & 10.42 & 76.73\% & 2.01 & 1.83 & 8.96\% \\
RNSGA-II & 26.44 & 3.07 & 88.39\% & 39.82 & 18.40 & 53.79\% & 16.15 & 0.71 & 95.60\% & 71.15 & 10.45 & 85.31\% & 11.17 & 1.83 & 83.62\% \\
SMS-EMOA & 28.62 & 2.83 & 90.11\% & 43.34 & 16.51 & 61.91\% & 16.71 & 0.69 & 95.87\% & 77.91 & 9.53 & 87.77\% & 9.77 & 1.83 & 81.27\% \\
MOEA/D-GEN & 35.44 & 4.81 & 86.43\% & 28.15 & 19.19 & 31.83\% & 14.81 & 2.31 & 84.40\% & 66.96 & 32.08 & 52.09\% & 3.79 & 1.83 & 51.72\% \\ \hline
\end{tabular}%
}
\end{table*}
\vspace{-16pt}
\end{center}

\begin{center}
\begin{table*}[ht]
\caption{Comparisons of Solution Sets from all algorithms in terms of IGD, $\Delta$, and Time}
\label{tb_comparison_metric}
\centering
\begin{tabular}{|c|ccc|ccc|ccc|ccc|ccc|}
\hline
Dataset & \multicolumn{3}{c|}{AGNEWS} & \multicolumn{3}{c|}{COQA} & \multicolumn{3}{c|}{HEADLINES} & \multicolumn{3}{c|}{SCIQ} & \multicolumn{3}{c|}{LogPai} \\ \hline
Metrics & IGD & $\Delta$ & Time & IGD & $\Delta$ & Time & IGD & $\Delta$ & Time & IGD & $\Delta$ & Time & IGD & $\Delta$ & Time \\ \hline
OptLLM & \textbf{0.13} & \textbf{0.55} & \textbf{6.15} & \textbf{0.27} & \textbf{0.79} & \textbf{7.64} & \textbf{0.08} & \textbf{0.53} & \textbf{7.71} & \textbf{0.14} & \textbf{0.58} & \textbf{10.01} & \textbf{0.19} & \textbf{0.95} & \textbf{18.14} \\
NSGA-II & 22.23 & 1.00 & 30.16 & 16.60 & 1.00 & 36.46 & 4.52 & 1.00 & 40.35 & 26.01 & 1.00 & 50.53 & 10.90 & 1.00 & 129.20 \\
MOPSO & 13.13 & 0.99 & 33.85 & 16.06 & 0.98 & 39.93 & 5.26 & 0.99 & 43.93 & 30.07 & 0.99 & 51.68 & 0.38 & 1.00 & 126.01 \\
MOEA/D & 13.97 & 0.99 & 33.91 & 23.26 & 1.00 & 39.65 & 4.44 & 0.99 & 43.79 & 44.32 & 1.00 & 50.97 & 1.05 & 0.98 & 127.24 \\
RNSGA-II & 23.85 & 0.99 & 30.43 & 34.87 & 0.99 & 35.10 & 15.02 & 0.99 & 42.78 & 63.54 & 0.99 & 50.76 & 10.17 & 1.00 & 128.11 \\
SMS-EMOA & 25.97 & 0.99 & 30.60 & 39.30 & 0.99 & 36.14 & 15.85 & 0.99 & 39.78 & 72.74 & 0.99 & 50.61 & 8.78 & 1.00 & 130.22 \\
MOEA/D-GEN & 9.12 & 1.09 & 35.06 & 23.91 & 1.16 & 39.54 & 4.69 & 1.07 & 44.80 & 42.52 & 1.13 & 45.47 & 0.75 & 1.07 & 124.00 \\ \hline
\end{tabular}
\end{table*}
\vspace{-28pt}
\end{center}

There are two possible reasons why OptLLM performs better. First, the solutions generated by the optimization component can be regarded as the ones with the best trade-off between cost and predicted accuracy (See Section~\ref{sec4OptLLM}). Second, multi-objective algorithms may be more prone to getting trapped in local optima, especially when the predicted accuracies for different queries are similar. In such cases, the algorithms may struggle to identify the globally optimal solutions without additional problem-specific information to guide the search process. OptLLM, on the other hand, incorporates heuristic rules through its prediction component, which helps to guide the optimization process and avoid local optima.

Overall, the experimental results demonstrate the superior performance of OptLLM in multiple aspects. 

To verify the comparison, we conduct a statistical test, and the results are provided on our project webpage due to the space limitation\textsuperscript{\ref{github}}.

\begin{mdframed}
\textit{\noindent In response to RQ1, our results show that OptLLM outperforms both single LLM allocation solutions and other multi-objective optimization strategies.}
\end{mdframed} 

\subsection{RQ2: Ablation Study}
\begin{figure}
    \centering
    \includegraphics[width=\linewidth]{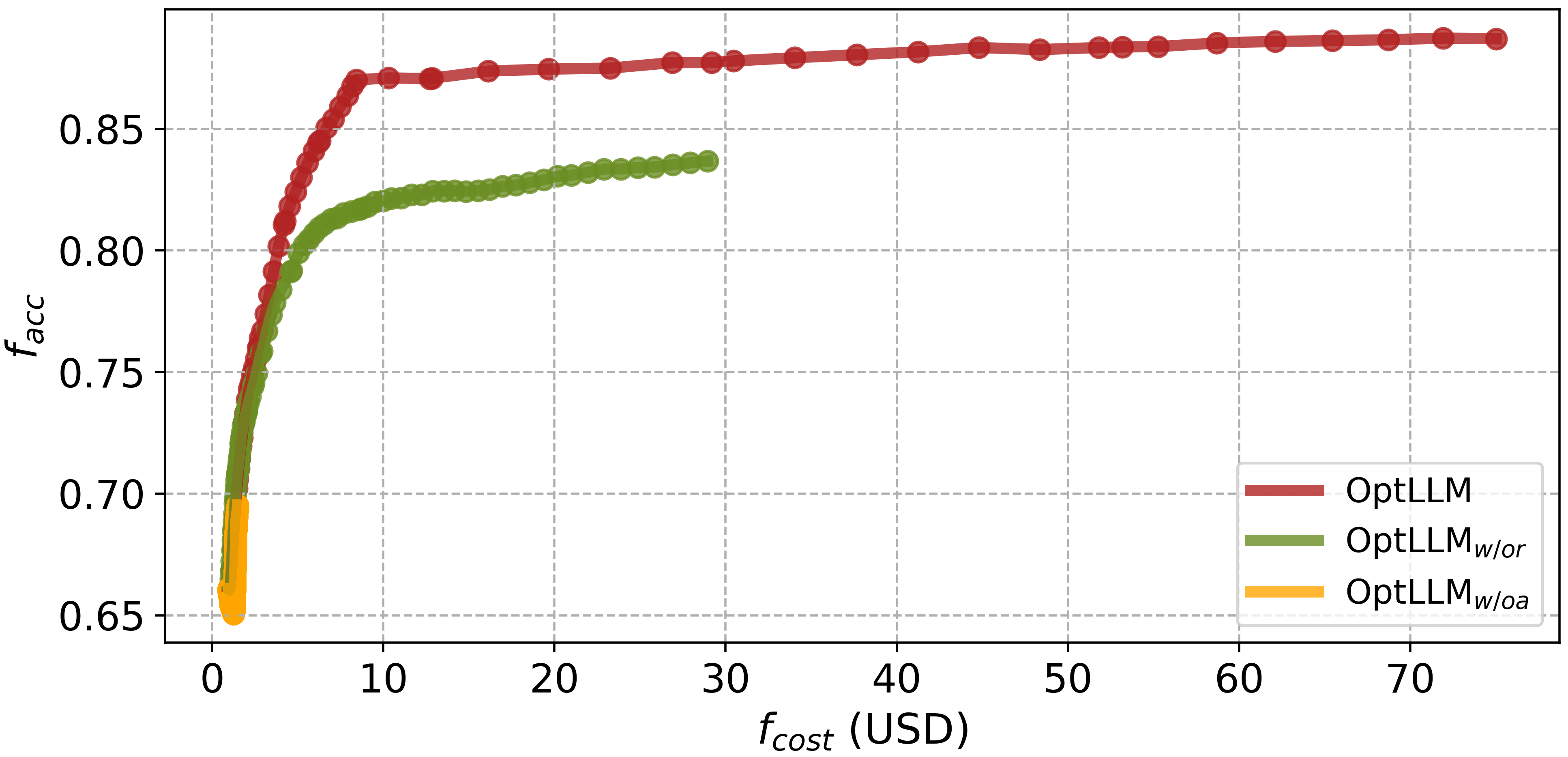}
    \caption{Ablation study on AGNEWS. The optimization component and robust-aware prediction help find the solution with higher accuracy. The solution with the highest accuracy by OptLLM$_{\text{w/o o}}$ is 69.54\%, OptLLM$_{\text{w/o r}}$ is 83.67\%, OptLLM is 88.74\%}
    \label{fig:ablation_agnews}
\end{figure}

To assess the performance improvement contributed by the optimization component of OptLLM, we conduct an ablation study. We design two variant algorithms: OptLLM$_{\text{w/o o}}$, which represents OptLLM without optimization, and OptLLM$_{\text{w/o r}}$, which represents OptLLM without robust-aware prediction. 
The results of OptLLM and its alternative versions are presented in Table~\ref{tb_ablation}.

In OptLLM$_{\text{w/o o}}$, we consider an LLM capable of processing a query correctly if the predicted probability of a correct response exceeds 0.5.
Solutions are then generated by assigning each job to the LLM predicted to process it correctly at the lowest cost. As illustrated in Figure~\ref{fig:ablation_agnews}, this approach successfully provides low-cost solutions but fails to identify solutions with high accuracy, indicating that the prediction component alone is insufficient for achieving optimal performance. The results presented in Table~\ref{tb_ablation} demonstrate the crucial role of the optimization module in OptLLM, as it significantly improves both the quality and diversity of the generated solutions. By facilitating the discovery of more high-accuracy solutions, the optimization process enhances the overall performance and versatility of the OptLLM framework. 

In OptLLM$_{\text{w/o r}}$, we use the predicted accuracy output by the base model to evaluate candidate solutions during the optimization phase. The difference between OptLLM$_{\text{w/o r}}$ and OptLLM is minimal. Further analysis reveals that when the prediction model has high accuracy, the performance of OptLLM$_{\text{w/o r}}$ is similar to that of OptLLM. This similarity is due to the smaller difference between the predicted accuracy and the robust-aware predicted accuracy when the prediction model performs well. In other words, the robust-aware function is useful when the predicted accuracy is suboptimal.
\begin{center}
\begin{table*}[ht]
\caption{Ablation study of OptLLM}
\label{tb_ablation}
\centering
\begin{tabular}{|c|ccc|ccc|ccc|ccc|ccc|}
\hline
 & \multicolumn{3}{c|}{AGNEWS} & \multicolumn{3}{c|}{COQA} & \multicolumn{3}{c|}{HEADLINES} & \multicolumn{3}{c|}{SCIQ} & \multicolumn{3}{c|}{LogPai} \\ \cline{2-16} 
Metrics & IGD & $\Delta$ & Time & IGD & $\Delta$ & Time & IGD & $\Delta$ & Time & IGD & $\Delta$ & Time & IGD & $\Delta$ & Time \\ \hline
OptLLM & \textbf{0.13} & \textbf{0.55} & \textbf{6.15} & 0.27 & \textbf{0.79} & \textbf{7.64} & \textbf{0.08} & \textbf{0.53} & \textbf{7.71} & \textbf{0.14} & \textbf{0.58} & \textbf{10.01} & \textbf{0.19} & \textbf{0.95} & \textbf{18.14} \\
$\text{OptLLM}_{\text{w/o r}}$ & 0.13 & 0.78 & 7.23 & 0.24 & 0.80 & 10.02 & 0.08 & 0.79 & 11.23 & 0.28 & 0.77 & 10.33 & 0.20 & 0.97 & 19.49 \\
$\text{OptLLM}_{\text{w/o a}}$& 0.23 & 0.86 & 14.38 & \textbf{0.12} & 1.1 & 16.74 & 0.1 & 0.87 &18.41  &0.40&1.08 &27.60 &0.19  & 1.21 & 500.32\\
\hline
\end{tabular}
\end{table*}
\vspace{-18pt}
\end{center}

\begin{center}
\begin{table*}[ht]
\caption{Comparison of OptLLM with different settings of GN (N: number of generated solutions)}
\label{tb_gn}
\centering
\resizebox{2\columnwidth}{!}{%
\begin{tabular}{|c|cccc|cccc|cccc|cccc|cccc|}
\hline
\multirow{2}{*}{Dataset} & \multicolumn{4}{c|}{AGNEWS} & \multicolumn{4}{c|}{COQA} & \multicolumn{4}{c|}{HEADLINES} & \multicolumn{4}{c|}{SCIQ} & \multicolumn{4}{c|}{LogPai} \\ \cline{2-21} 
 & IGD & $\Delta$ & Time & N & IGD & $\Delta$ & Time & N & IGD & $\Delta$ & Time & N & IGD & $\Delta$ & Time & N & IGD & $\Delta$ & Time & N \\ \hline
GN=10 & 0.17 & 0.82 & \textbf{2.29} & 14 & \textbf{0.27} & \textbf{0.80} & \textbf{2.13} & 20 & 0.10 & 0.85 & \textbf{3.49} & 16 & 0.18 & 0.80 & \textbf{4.81} & 18 & \textbf{0.14} & \textbf{0.83} & \textbf{7.78} & 20 \\
GN=50 & 0.13 & 0.80 & 6.35 & 71 & 0.27 & 0.80 & 7.64 & 99 & 0.08 & 0.79 & 8.85 & 81 & 0.14 & 0.77 & 10.01 & 92 & 0.19 & 0.95 & 18.14 & 100 \\
GN=100 & 0.13 & 0.80 & 11.82 & 144 & 0.30 & 0.81 & 13.52 & 199 & 0.08 & 0.78 & 15.65 & 162 & 0.14 & 0.77 & 20.45 & 189 & 0.20 & 0.96 & 37.87 & 200 \\
GN=200 & \textbf{0.13} & \textbf{0.79} & 18.25 & \textbf{291} & 0.33 & 0.82 & 23.42 & \textbf{398} & \textbf{0.08} & \textbf{0.77} & 25.35 & \textbf{327} & \textbf{0.13} & \textbf{0.77} & 35.72 & \textbf{387} & 0.18 & 0.95 & 63.60 & \textbf{400} \\ \hline
\end{tabular}%
}
\end{table*}
\vspace{-18pt}
\end{center}

\begin{center}
\begin{table}[ht]
\label{tbl:train_size}
\caption{Results under different training data sizes}
\resizebox{\columnwidth}{!}{%
\begin{tabular}{|c|cc|cc|cc|cc|cc|}
\hline
Dataset & \multicolumn{2}{c|}{AGNEWS} & \multicolumn{2}{c|}{COQA} & \multicolumn{2}{c|}{HEADLINES} & \multicolumn{2}{c|}{SCIQ} & \multicolumn{2}{c|}{LogPai} \\ \hline
size & IGD & $\Delta$ & IGD & $\Delta$ & IGD & $\Delta$ & IGD & $\Delta$ & IGD & $\Delta$ \\ \hline
1\% & 0.13 & \textbf{0.65} & 0.23 & \textbf{0.66} & 0.08 & \textbf{0.61} & \textbf{0.13} & \textbf{0.61} & 0.19 & \textbf{0.87} \\
5\% & 0.11 & 0.71 & 0.19 & 0.67 & 0.07 & 0.71 & 0.14 & 0.63 & 0.18 & 1.01 \\
10\% & 0.10 & 0.73 & 0.17 & 0.69 & 0.06 & 0.74 & 0.14 & 0.63 & 0.18 & 1.09 \\
20\% & \textbf{0.09} & 0.74 & \textbf{0.15} & 0.68 & \textbf{0.06} & 0.73 & 0.14 & 0.63 & \textbf{0.16} & 1.10 \\ \hline
\end{tabular}%
}
\end{table}
\vspace{-18pt}
\end{center}

\begin{center}
\begin{table}[ht]
\caption{Accuracy of prediction model with different $\alpha$}
\label{tbl_alpha}
\resizebox{\columnwidth}{!}{%
\begin{tabular}{|l|c|c|c|c|c|}
\hline
\multicolumn{1}{|c|}{Dataset} & AGNEWS & COQA & HEADLINES & SCIQ & LogPai \\ \hline
base model & 0.21 & 0.45 & 0.44 & 0.25 & \textbf{0.46} \\
$\alpha=-1$ & 0.14 & \textbf{0.50} & 0.41 & 0.17 & 0.39 \\
$\alpha=-0.5$ & 0.19 & 0.48 & 0.45 & 0.24 & 0.43 \\
$\alpha=0$ & 0.24 & 0.47 & 0.46 & 0.31 & 0.45 \\
$\alpha=0.5$ & 0.26 & 0.43 & 0.47 & 0.35 & 0.43 \\
$\alpha=1$ & \textbf{0.27} & 0.40 & \textbf{0.48} & \textbf{0.37} & 0.38 \\ \hline
\end{tabular}%
}
\end{table}
\vspace{-18pt}
\end{center}
Figure~\ref{fig:ablation_agnews} provides an example using the AGNEWS dataset. It clearly shows that OptLLM$_{\text{w/o o}}$ generates substantially fewer solutions and fails to find solutions with higher accuracy. Although OptLLM$_{\text{w/o r}}$ performs better than OptLLM$_{\text{w/o o}}$, OptLLM, which incorporates a robust-aware function, excels in searching for solutions with higher accuracy.

\begin{mdframed}
\textit{\noindent In response to RQ2, both robust-aware prediction function and heuristic optimization significantly enhance the overall capability 
of the OptLLM.
}
\end{mdframed}

\subsection{RQ3: Effect of Hyper-Parameter Settings}

\subsubsection{The grid parameter GN}
In Table~\ref{tb_gn}, we show IGD, $\Delta$, average runtime, and average size of the solution set produced by OptLLM under various settings of the GN parameter, which controls objective distance during the optimization component's destruction phase. Analysis of the table shows that increasing the GN value leads to a larger number of solutions generated by OptLLM but at the cost of increased computation time. While the solution quality, as indicated by IGD and $\Delta$, remains similar across different GN settings.

\subsubsection{The training data size of prediction component}
Table~\ref{tbl:train_size} presents the results of OptLLM under different training data sizes for the prediction component. As the size of the training data increases, the $\Delta$ metric consistently increases, while the IGD shows a slight decrease overall. 
Consequently, we choose to use 1\% of the dataset for training, striking a balance between performance and the cost of acquiring labelled data.

\subsubsection{The robustness parameter $\alpha$}
Table \ref{tbl_alpha} presents the accuracy of the predictions on different datasets when using different values of the robustness parameter $\alpha$. The base model represents the accuracy of the classification model trained directly on the training data without using bootstrap sampling. The impact of the robustness parameter $\alpha$ on prediction accuracy varies across datasets. While some datasets (AGNEWS, HEADLINES, SCIQ) benefit from higher values of $\alpha$, others (COQA, LogPai) achieve better performance with lower or zero values of $\alpha$.
\begin{mdframed}
\textit{\noindent In response to RQ3, OptLLM allows for the adjustment of GN to a larger value to generate more alternative solutions, while it prefers a relatively small training size for decreased $\Delta$ and lower costs in labelled data acquisition. Moreover, the effect of the robustness parameter $\alpha$ on prediction accuracy is found to differ among datasets.
}
\end{mdframed}

\section{Related work}
LLMs have been adopted to address a variety of tasks and have shown numerous potential breakthroughs~\cite{le2023log, jiang2023llmparser, hou2023large, sobania2023analysis}. This has driven the demand for optimizing the performance and cost of LLM-based tasks. Zong et al.\cite{zong2023model} conducted an empirical study assessing the annotation, training, and inference costs of various LLMs in text classification, offering guidance for selecting the most suitable model in real-world scenarios. Chen et al.\cite{chen2023frugalgpt} introduced FrugalGPT, an algorithmic framework that adaptively selects suitable LLMs for different queries to reduce cost and improve accuracy. Specifically, it sends queries sequentially to available LLMs until reaching a predefined performance threshold defined by a scoring function. Similarly, Sakota et al.~\cite{vsakota2023fly} developed FORC, a framework that allocates queries to suitable LLMs by predicting cost and performance of each candidate LLM on every query. Notably, these prediction models often require substantial training data (e.g., 50\%), incurring high label collection cost. Unlike the above work, we formulate the problem as a multi-objective optimization and use less data (i.e., 1\%) to construct the prediction model. We focus on the one-time allocation of each query to a suitable LLM, aiming to achieve a trade-off between cost and accuracy.

\section{Conclusion}
\label{sec8:Conclusion}
The substantial cost associated with leveraging LLMs in real-world scenarios poses a significant obstacle to their widespread adoption. In this paper, we propose OptLLM, an effective and efficient framework that automatically assigns queries to suitable LLMs. OptLLM offers a diverse set of non-dominated solutions, from which users can select based on their budget and performance requirements. These solutions span the range from the highest expected accuracy to the lowest cost alternatives. Our experimental results demonstrate that OptLLM outperforms other baseline methods in terms of the effectiveness and efficiency of query assignments. 

The source code of OptLLM and all experimental results are available at 
\url{https://github.com/superyue72/OptLLM}.

\balance
\bibliographystyle{IEEEtran}
\bibliography{reference}

\begin{thebibliography}{10}
\providecommand{\url}[1]{#1}
\csname url@samestyle\endcsname
\providecommand{\newblock}{\relax}
\providecommand{\bibinfo}[2]{#2}
\providecommand{\BIBentrySTDinterwordspacing}{\spaceskip=0pt\relax}
\providecommand{\BIBentryALTinterwordstretchfactor}{4}
\providecommand{\BIBentryALTinterwordspacing}{\spaceskip=\fontdimen2\font plus
\BIBentryALTinterwordstretchfactor\fontdimen3\font minus \fontdimen4\font\relax}
\providecommand{\BIBforeignlanguage}[2]{{%
\expandafter\ifx\csname l@#1\endcsname\relax
\typeout{** WARNING: IEEEtran.bst: No hyphenation pattern has been}%
\typeout{** loaded for the language `#1'. Using the pattern for}%
\typeout{** the default language instead.}%
\else
\language=\csname l@#1\endcsname
\fi
#2}}
\providecommand{\BIBdecl}{\relax}
\BIBdecl

\bibitem{kasneci2023chatgpt}
E.~Kasneci, K.~Se{\ss}ler, S.~K{\"u}chemann, M.~Bannert, D.~Dementieva, F.~Fischer, U.~Gasser, G.~Groh, S.~G{\"u}nnemann, E.~H{\"u}llermeier \emph{et~al.}, ``Chatgpt for good? on opportunities and challenges of large language models for education,'' \emph{Learning and individual differences}, vol. 103, p. 102274, 2023.

\bibitem{ouyang2022training}
L.~Ouyang, J.~Wu, X.~Jiang, D.~Almeida, C.~Wainwright, P.~Mishkin, C.~Zhang, S.~Agarwal, K.~Slama, A.~Ray \emph{et~al.}, ``Training language models to follow instructions with human feedback,'' \emph{Advances in Neural Information Processing Systems}, vol.~35, pp. 27\,730--27\,744, 2022.

\bibitem{xie2022explanation}
S.~M. Xie, A.~Raghunathan, P.~Liang, and T.~Ma, ``An explanation of in-context learning as implicit bayesian inference,'' \emph{arXiv preprint arXiv:2111.02080}, 2021.

\bibitem{min2022rethinking}
S.~Min, X.~Lyu, A.~Holtzman, M.~Artetxe, M.~Lewis, H.~Hajishirzi, and L.~Zettlemoyer, ``Rethinking the role of demonstrations: What makes in-context learning work?'' in \emph{Proceedings of the 2022 Conference on Empirical Methods in Natural Language Processing (EMNLP)}, 2022, pp. 11\,048--11\,064.

\bibitem{liu2023your}
J.~Liu, C.~S. Xia, Y.~Wang, and L.~Zhang, ``Is your code generated by chatgpt really correct? rigorous evaluation of large language models for code generation,'' \emph{arXiv preprint arXiv:2305.01210}, 2023.

\bibitem{xia2023automated}
C.~S. Xia, Y.~Wei, and L.~Zhang, ``Automated program repair in the era of large pre-trained language models,'' in \emph{Proceedings of the 45th International Conference on Software Engineering (ICSE)}, 2023.

\bibitem{xia2023keep}
C.~S. Xia and L.~Zhang, ``Keep the conversation going: Fixing 162 out of 337 bugs for \$0.42 each using chatgpt,'' \emph{arXiv preprint arXiv:2304.00385}, 2023.

\bibitem{vsakota2023fly}
M.~{\v{S}}akota, M.~Peyrard, and R.~West, ``Fly-swat or cannon? cost-effective language model choice via meta-modeling,'' \emph{arXiv preprint arXiv:2308.06077}, 2023.

\bibitem{gpt3}
\BIBentryALTinterwordspacing
(2023) How much does it cost to use gpt models? gpt-3 pricing explained. [Online]. Available: \url{https://neoteric.eu/blog/how-much-does-it-cost-to-use-gpt-models-gpt-3-pricing-explained/}
\BIBentrySTDinterwordspacing

\bibitem{chen2023frugalgpt}
L.~Chen, M.~Zaharia, and J.~Zou, ``Frugalgpt: How to use large language models while reducing cost and improving performance,'' \emph{arXiv preprint arXiv:2305.05176}, 2023.

\bibitem{jiang2024mixtral}
A.~Q. Jiang, A.~Sablayrolles, A.~Roux, A.~Mensch, B.~Savary, C.~Bamford, D.~S. Chaplot, D.~d.~l. Casas, E.~B. Hanna, F.~Bressand \emph{et~al.}, ``Mixtral of experts,'' \emph{arXiv preprint arXiv:2401.04088}, 2024.

\bibitem{wu2023its}
H.~Wu, M.~Wu, W.~Peng, S.~Chen, and Z.~Feng, ``Its: Improved tabu search algorithm for path planning in uav-assisted edge computing systems,'' in \emph{Proceedings of the 2023 IEEE International Conference on Web Services (ICWS)}.\hskip 1em plus 0.5em minus 0.4em\relax IEEE, 2023, pp. 340--349.

\bibitem{haq2022efficient}
F.~U. Haq, D.~Shin, and L.~Briand, ``Efficient online testing for dnn-enabled systems using surrogate-assisted and many-objective optimization,'' in \emph{Proceedings of the 44th international conference on software engineering (ICSE)}, 2022, pp. 811--822.

\bibitem{cheng2016test}
R.~Cheng, Y.~Jin, M.~Olhofer \emph{et~al.}, ``Test problems for large-scale multiobjective and many-objective optimization,'' \emph{IEEE transactions on cybernetics}, vol.~47, no.~12, pp. 4108--4121, 2016.

\bibitem{he2020adaptive}
C.~He, R.~Cheng, and D.~Yazdani, ``Adaptive offspring generation for evolutionary large-scale multiobjective optimization,'' \emph{IEEE Transactions on Systems, Man, and Cybernetics: Systems}, vol.~52, no.~2, pp. 786--798, 2020.

\bibitem{openai2023gpt4}
OpenAI, ``Gpt-4 technical report,'' 2023.

\bibitem{touvron2023llama}
H.~Touvron, T.~Lavril, G.~Izacard, X.~Martinet, M.-A. Lachaux, T.~Lacroix, B.~Rozi{\`e}re, N.~Goyal, E.~Hambro, F.~Azhar \emph{et~al.}, ``Llama: Open and efficient foundation language models,'' \emph{arXiv preprint arXiv:2302.13971}, 2023.

\bibitem{le2023log}
V.-H. Le and H.~Zhang, ``Log parsing: How far can chatgpt go?'' in \emph{Proceedings of the 38th (2023) IEEE/ACM International Conference on Automated Software Engineering (ASE)}.\hskip 1em plus 0.5em minus 0.4em\relax IEEE, 2023.

\bibitem{deb2016multi}
K.~Deb, K.~Sindhya, and J.~Hakanen, ``Multi-objective optimization,'' in \emph{Decision sciences}.\hskip 1em plus 0.5em minus 0.4em\relax CRC Press, 2016, pp. 161--200.

\bibitem{ramirez2019survey}
A.~Ramirez, J.~R. Romero, and S.~Ventura, ``A survey of many-objective optimisation in search-based software engineering,'' \emph{Journal of Systems and Software}, vol. 149, pp. 382--395, 2019.

\bibitem{cheikh2010method}
M.~Cheikh, B.~Jarboui, T.~Loukil, and P.~Siarry, ``A method for selecting pareto optimal solutions in multiobjective optimization,'' \emph{Journal of Informatics and mathematical sciences}, vol.~2, no.~1, pp. 51--62, 2010.

\bibitem{konak2006multi}
A.~Konak, D.~W. Coit, and A.~E. Smith, ``Multi-objective optimization using genetic algorithms: A tutorial,'' \emph{Reliability engineering \& system safety}, vol.~91, no.~9, pp. 992--1007, 2006.

\bibitem{ben2009robust}
A.~Ben-Tal, L.~El~Ghaoui, and A.~Nemirovski, \emph{Robust optimization}.\hskip 1em plus 0.5em minus 0.4em\relax Princeton university press, 2009, vol.~28.

\bibitem{zhang2015character}
X.~Zhang, J.~Zhao, and Y.~LeCun, ``Character-level convolutional networks for text classification,'' \emph{Advances in neural information processing systems}, vol.~28, 2015.

\bibitem{reddy2019coqa}
S.~Reddy, D.~Chen, and C.~D. Manning, ``Coqa: A conversational question answering challenge,'' \emph{Transactions of the Association for Computational Linguistics}, vol.~7, pp. 249--266, 2019.

\bibitem{sinha2021impact}
A.~Sinha and T.~Khandait, ``Impact of news on the commodity market: Dataset and results,'' in \emph{Advances in Information and Communication: Proceedings of the 2021 Future of Information and Communication Conference (FICC), Volume 2}.\hskip 1em plus 0.5em minus 0.4em\relax Springer, 2021, pp. 589--601.

\bibitem{welbl2017crowdsourcing}
J.~Welbl, N.~F. Liu, and M.~Gardner, ``Crowdsourcing multiple choice science questions,'' \emph{arXiv preprint arXiv:1707.06209}, 2017.

\bibitem{zhu2019tools}
J.~Zhu, S.~He, J.~Liu, P.~He, Q.~Xie, Z.~Zheng, and M.~R. Lyu, ``Tools and benchmarks for automated log parsing,'' in \emph{Proceedings of the 2019 IEEE/ACM 41st International Conference on Software Engineering: Software Engineering in Practice (ICSE-SEIP)}.\hskip 1em plus 0.5em minus 0.4em\relax IEEE, 2019, pp. 121--130.

\bibitem{khan2022guidelines}
Z.~A. Khan, D.~Shin, D.~Bianculli, and L.~Briand, ``Guidelines for assessing the accuracy of log message template identification techniques,'' in \emph{Proceedings of the 44th International Conference on Software Engineering (ICSE)}, 2022, pp. 1095--1106.

\bibitem{deb2002fast}
K.~Deb, A.~Pratap, S.~Agarwal, and T.~Meyarivan, ``A fast and elitist multiobjective genetic algorithm: Nsga-ii,'' \emph{IEEE transactions on evolutionary computation}, vol.~6, no.~2, pp. 182--197, 2002.

\bibitem{coello2002mopso}
C.~C. Coello and M.~S. Lechuga, ``Mopso: A proposal for multiple objective particle swarm optimization,'' in \emph{Proceedings of the 2002 Congress on Evolutionary Computation (CEC)}, vol.~2.\hskip 1em plus 0.5em minus 0.4em\relax IEEE, 2002, pp. 1051--1056.

\bibitem{zhang2007moea}
Q.~Zhang and H.~Li, ``Moea/d: A multiobjective evolutionary algorithm based on decomposition,'' \emph{IEEE Transactions on evolutionary computation}, vol.~11, no.~6, pp. 712--731, 2007.

\bibitem{chen2018beyond}
J.~Chen, V.~Nair, and T.~Menzies, ``Beyond evolutionary algorithms for search-based software engineering,'' \emph{Information and Software Technology}, vol.~95, pp. 281--294, 2018.

\bibitem{deb2006reference}
K.~Deb and J.~Sundar, ``Reference point based multi-objective optimization using evolutionary algorithms,'' in \emph{Proceedings of the 8th Annual Conference on Genetic and Evolutionary Computation (GECCO)}, 2006, pp. 635--642.

\bibitem{beume2007sms}
N.~Beume, B.~Naujoks, and M.~Emmerich, ``Sms-emoa: Multiobjective selection based on dominated hypervolume,'' \emph{European Journal of Operational Research}, vol. 181, no.~3, pp. 1653--1669, 2007.

\bibitem{wang2018scalable}
Z.~Wang, Y.-S. Ong, and H.~Ishibuchi, ``On scalable multiobjective test problems with hardly dominated boundaries,'' \emph{IEEE Transactions on Evolutionary Computation}, vol.~23, no.~2, pp. 217--231, 2018.

\bibitem{audet2021performance}
C.~Audet, J.~Bigeon, D.~Cartier, S.~Le~Digabel, and L.~Salomon, ``Performance indicators in multiobjective optimization,'' \emph{European journal of operational research}, vol. 292, no.~2, pp. 397--422, 2021.

\bibitem{pymoo}
J.~{Blank} and K.~{Deb}, ``pymoo: Multi-objective optimization in python,'' \emph{IEEE Access}, vol.~8, pp. 89\,497--89\,509, 2020.

\bibitem{akiba2019optuna}
T.~Akiba, S.~Sano, T.~Yanase, T.~Ohta, and M.~Koyama, ``Optuna: A next-generation hyperparameter optimization framework,'' in \emph{Proceedings of the 25th ACM SIGKDD International Conference on Knowledge Discovery \& Data Mining (KDD)}, 2019, pp. 2623--2631.

\bibitem{jiang2023llmparser}
Z.~Jiang, J.~Liu, Z.~Chen, Y.~Li, J.~Huang, Y.~Huo, P.~He, J.~Gu, and M.~R. Lyu, ``Llmparser: A llm-based log parsing framework,'' \emph{arXiv preprint arXiv:2310.01796}, 2023.

\bibitem{hou2023large}
X.~Hou, Y.~Zhao, Y.~Liu, Z.~Yang, K.~Wang, L.~Li, X.~Luo, D.~Lo, J.~Grundy, and H.~Wang, ``Large language models for software engineering: A systematic literature review,'' \emph{arXiv preprint arXiv:2308.10620}, 2023.

\bibitem{sobania2023analysis}
D.~Sobania, M.~Briesch, C.~Hanna, and J.~Petke, ``An analysis of the automatic bug fixing performance of chatgpt,'' \emph{arXiv preprint arXiv:2301.08653}, 2023.

\bibitem{zong2023model}
S.~Zong, J.~Seltzer, K.~Cheng, J.~Lin \emph{et~al.}, ``Which model shall i choose? cost/quality trade-offs for text classification tasks,'' \emph{arXiv preprint arXiv:2301.07006}, 2023.

\end{thebibliography}

\balance

\end{document}